\documentclass[a4paper,UKenglish,cleveref, autoref, thm-restate,authorcolumns]{lipics-v2019}

\title{\(\alpha\)-approximate Reductions: a Novel Source of Heuristics for
  Better Approximation Algorithms}

\titlerunning{\(\alpha\)-approximate Reductions as Good Heuristics for Approximation } 

\author{Fredrik Manne}{University of Bergen, Bergen, Norway}{Fredrik.Manne@uib.no}{}{}

\author{Geevarghese Philip}{Chennai Mathematical Institute, India}{gphilip@cmi.ac.in}{}{}

\author{Saket  Saurabh}{The Institute of Mathematical Sciences, HBNI, Chennai, India \\
University of Bergen, Bergen, Norway
}{saket@imsc.res.in}{}{}
\author{Prafullkumar Tale}{CISPA Helmholtz Center for Information Security, Saarbr$\ddot{\text{u}}$cken, Germany}{prafullkumar.tale@cispa.saarland}{}{}

\authorrunning{Manne, Philip, Saurabh, and Tale} 

\Copyright{Fredrik Manne, Geevarghese Philip, Saket  Saurabh, and Prafullkumar Tale} 


\category{} 

\relatedversion{} 

\supplement{{\em Code Repository}. URL: \url{https://github.com/pptale/DomSet-Lossy}}

\EventEditors{John Q. Open and Joan R. Access}
\EventNoEds{2}
\EventLongTitle{42nd Conference on Very Important Topics (CVIT 2016)}
\EventShortTitle{CVIT 2016}
\EventAcronym{CVIT}
\EventYear{2016}
\EventDate{December 24--27, 2016}
\EventLocation{Little Whinging, United Kingdom}
\EventLogo{}
\SeriesVolume{42}
\ArticleNo{23}

\usepackage[linesnumbered,algoruled,boxed,lined]{algorithm2e}
\usepackage{complexity} 
\usepackage{todonotes} 
\presetkeys{todonotes}{inline}{}
\usepackage{xspace}
\newcommand{\Cplusplus}{\texttt{C}\raisebox{0.5ex}{\tiny\texttt{++}}\xspace}

\newcommand{\DomSet}{\textsc{Dominating Set}\xspace}

\newcommand{\DS}{\textsc{DS}\xspace}
\newtheorem{solution lifting}{Solution~Lifting~Algorithm\xspace}[section]

\newcommand{\calO}{\ensuremath{{\mathcal O}}}

\newcommand{\OPT}{\textsc{OPT}}

\newcommand{\yes}{\textsc{Yes}\xspace}

\newcommand{\YES}{\textsc{Yes}\xspace}
\newcommand{\NO}{\textsc{No}\xspace}

\newtheorem{reduction rule}{Reduction Rule}[section]
\newtheorem{marking-scheme}{Marking Scheme}[section]

\begin{document}

\maketitle

\begin{abstract}
Lokshtanov et al.~[STOC 2017] introduced \emph{lossy kernelization} as a
mathematical framework for quantifying the effectiveness of preprocessing
algorithms in preserving approximation ratios. \emph{\(\alpha\)-approximate
  reduction rules} are a central notion of this framework. We propose that
carefully crafted \(\alpha\)-approximate reduction rules can yield improved
approximation ratios in practice, while being easy to implement as well. This is
distinctly different from the (theoretical) purpose for which Lokshtanov et
al. designed \(\alpha\)-approximate Reduction Rules. As evidence in support of
this proposal we present a new \(2\)-approximate reduction rule for the \DomSet
problem.  This rule, when combined with an approximation algorithm for \DomSet,
yields significantly better approximation ratios on a variety of benchmark
instances as compared to the latter algorithm alone.

  The central thesis of this work is that \(\alpha\)-approximate reduction rules
  can be used as a tool for designing approximation algorithms which perform
  better in practice. To the best of our knowledge, ours is the first
  exploration of the use of \(\alpha\)-approximate reduction rules as a design
  technique for practical approximation algorithms. We believe that this
  technique could be useful in coming up with improved approximation algorithms
  for other optimization problems as well.


\end{abstract}
\newpage
\setcounter{page}{1}

\section{Introduction}
The notion of \emph{kernelization} as defined in the field of parameterized
algorithms is a mathematical tool for quantifying the effectiveness of
polynomial-time preprocessing algorithms in dealing with decision
problems~\cite{DF-new,fomin2019kernelization}. An instance of a parameterized
problem is of the form \((I, k)\) where \(I\) is an instance of the
underlying---usually, but not necessarily, \NP-hard---decision problem and \(k\)
is a---usually, but not necessarily, numerical---\emph{parameter}. The informal
idea is that the parameter \(k\) captures some aspect which makes the instance
\(I\) hard to solve. A \emph{kernelization algorithm} for this parameterized
problem takes \((I, k)\) as input, runs in time polynomial in \((|I| + k)\), and
outputs an instance \((I', k')\) of the same parameterized problem. Further,
\(I'\) is a \yes instance of the decision problem if and only if \(I\) is a \yes
instance, and the size \((|I'| + k')\) of \((I', k')\) is upper-bounded by some
function \(f(k)\) of the parameter \(k\) alone. \((I', k')\) is then a
\emph{kernel} of size at most \(f(k)\). A kernelization algorithm is thus a
polynomial-time preprocessing algorithm with a guaranteed upper bound on output
size, where the guarantee comes in terms of a function of the chosen parameter
\(k\).

The notion of kernelization has turned out to be very prescient and fruitful;
useful and non-trivial upper and lower bounds on kernel sizes have been derived
for literally hundreds of decision problems (See, e.g.,
\cite{fomin2019kernelization}). Nonetheless, kernelization has one significant
drawback when it comes to judging the effectiveness of preprocessing algorithms.
In general, a kernelization algorithm gives no guarantees as to how (or whether)
the quality of \emph{approximate solutions} for a kernel relate to the quality
of approximate solutions for the original input instance. Kernelization mandates
only that the reduced instance (i) be of bounded size, and (ii) preserve the
\YES/\NO answer to the original decision question. It gives no guarantees that a
``good'' solution for the kernel, obtained by running an approximation algorithm
or a heuristic on the kernel, can be converted to a ``good'' solution for the
original instance. This is a serious handicap in practice because in many cases
the output of preprocessing algorithms are still too large to be solved exactly,
and so must be processed further by approximation algorithms or heuristics.
Using kernelization to do the preprocessing may then prevent us from
getting---with any sort of approximation guarantee---good solutions to the
actual input instance, from good solutions to the kernel.

Lokshtanov et al.~\cite{LokshtanovPRS17-lossy} introduced \emph{lossy
  kernelization} as a theoretical framework for studying those polynomial-time
preprocessing algorithms which give a two-fold guarantee: (i) on the size of the
reduced instance, and (ii) on the loss in approximation factor with which one
can convert a solution of a reduced instance to a solution of the input
instance. Such algorithms are called \emph{\(\alpha\)-approximate kernelization
  algorithms} or \emph{\(\alpha\)-approximate kernels}. We now give an informal
description of these algorithms; please see \autoref{subsec:lossy} for precise
definitions. For a real number \(\alpha \geq 1\), an
\emph{\(\alpha\)-approximate reduction rule
} for a parameterized optimization problem consists of a \emph{pair} of
algorithms, a \emph{reduction algorithm} and a \emph{solution-lifting
  algorithm}. The reduction algorithm takes an instance \((I, k)\) of the
problem as input, runs in time polynomial in \((|I| + k)\), and outputs an
instance \((I', k')\) of the problem. The solution-lifting algorithm takes the
two instances \((I, k), (I', k')\) \emph{and} an arbitrary solution \(S'\) to
\((I', k')\) and constructs a solution \(S\) to \((I, k)\) with the following
guarantee: If \(\beta \geq 1\) is the approximation factor of \(S'\) with
respect to \((I', k')\) then the approximation factor of \(S\) with respect to
\((I, k)\) is \emph{no worse than} \(\alpha \times \beta\). That is, if we
manage to get hold of a ``good'' approximate solution to the reduced instance
\((I', k')\) then the solution-lifting algorithm will give us a solution to the
original instance which is not much worse in terms of the respective
approximation factors. An \emph{\(\alpha\)-approximate kernel} is an
\(\alpha\)-approximate
reduction rule for which the size \((|I'| + k')\) of the reduced instance is
upper-bounded by some function of the parameter \(k\).

In their pioneering work Lokshtanov et al. showed, \emph{inter alia}, that
various parameterized problems which do not\footnote{Under standard
  complexity-theoretic assumptions.} admit ``classical'' polynomial kernels have
\(\alpha\)-approximate kernels (also called ``lossy kernels'') of polynomial
size. This line of work has attracted much attention, and finding lossy kernels
of polynomial size, especially for problems which don't admit classical kernels
of polynomial size, is an important---and challenging---area of current research
in parameterized
algorithms~\cite{krithika2016lossy,eiben2017lossy,dvovrak2018parameterized,eiben2019lossy}.

\subparagraph*{Our Work.} We posit that the notion of \(\alpha\)-approximate
reduction rules has important practical utility beyond, and somewhat orthogonal
to, their original theoretical \emph{raison d'\^{e}tre} which we have outlined
above. Specifically, we claim that this notion can be used as \emph{an aid for
  designing approximation algorithms} which work well \emph{in practice}. That
is, we propose that starting out by designing an \(\alpha\)-approximate
reduction rule is a good way to arrive at a heuristic which is easy to implement
and has good approximation characteristics. The intuition behind this claim is
that the various conditions which---by definition---an \(\alpha\)-approximate
reduction rule should satisfy have the combined effect that such a rule is
likely to have these desirable qualities.

Needless to say, this claim is far from being well-defined. So we can only
provide some evidence, and not a proof, in support of the claim. As evidence we
present a \(2\)-approximate reduction rule for the well-studied \DomSet problem,
and show by empirical means that this rule is a good approximation heuristic for
\DomSet. We show that when combined with a state-of-the-art approximation
algorithm for \DomSet, the new \(2\)-approximate reduction rule gives improved
approximation ratios for a variety of benchmark instances drawn from various
sources.

\subparagraph*{Our Results.} Our main theoretical contribution is the new
\(2\)-approximate reduction rule for \DomSet, along with the proof that the rule
is indeed \(2\)-approximate. Our main practical contribution is a new heuristic
approximation algorithm for \DomSet. At a high level, this algorithm works as
follows. It starts by exhaustively applying two (folklore) \(1\)-approximate
reduction algorithms. It then applies the new \(2\)-approximate reduction
algorithm once to the resulting instance, to get a reduced instance. It then
runs a known approximation algorithm (the ``drop-in'' approximation algorithm)
for \DomSet on the reduced instance to get an approximate solution for the
reduced instance. Finally, it applies the various solution-lifting algorithms in
the right order to get an approximate solution to the input instance.

We tested this new heuristic algorithm on various well-known benchmark
instances. Since all these benchmark instances consisted of \emph{sparse}
graphs, we used the approximation algorithm of Jones et
al.~\cite{jones2017parameterized} for \DomSet on sparse graphs as the drop-in
approximation algorithm. As we describe in \autoref{sec:results}, our
three-stage algorithm---first apply the various reductions, then apply the
drop-in approximation, and finally apply the lifting algorithms---gives
significantly better results on a good fraction of the benchmark instances, as
compared to applying just the drop-in algorithm alone to these instances.

\subparagraph*{Organization of the rest of the paper.} In \autoref{subsec:graph}
we list various graph-theoretic preliminaries, and in \autoref{subsec:lossy} we
give a concise description of those concepts from the lossy kernels framework
which we need in the rest of the paper. In \autoref{sec:lossy} we describe the
two folklore reduction rules for \DomSet and prove that each of them is
\(1\)-approximate. In the same section we state our new reduction rule for
\DomSet and prove that it is \(2\)-approximate. In \autoref{sec:exp-setup} we
describe our experimental setup, and in \autoref{sec:results} we describe our
algorithm and the results we obtained on the various benchmark instances. We
conclude in \autoref{sec:conclusion}.


\section{Preliminaries}
\label{sec:prelims}

For positive constants $x, y, z, w$, we have $\frac{x + y}{z + w} \le \max\{\frac{x}{z}, \frac{y}{w}\}$.

\subsection{Graph Theory}
\label{subsec:graph}

All our graphs are undirected and simple. For a graph $G$ we use $V(G)$ and
$E(G)$ to denote the set of vertices and edges of \(G\), respectively. Unless
specified otherwise, we use $n, m$ to denote the cardinalities of sets $V(G)$
and $E(G)$, respectively. Two vertices $u, v$ are said to be \emph{adjacent} if
there is an edge $uv$ in the graph. The \emph{neighborhood} of a vertex $v$,
denoted by $N_G(v)$, is the set of vertices adjacent to $v$ and its
\emph{degree} $deg_G(v)$ is $|N_G(v)|$. The \emph{closed neighborhood} of a
vertex $v$, denoted $N_G[v]$, is $\{v\} \cup N(v)$. The subscript in the
notation for neighborhood and degree is omitted if the graph being discussed is
clear. For a subset $S$ of $V(G)$ we define
$N(S) = (\bigcup_{v \in S} N(S) ) \setminus S$ and
$N[S] = \bigcup_{v \in S} N[s]$. For set $S \subseteq V(G)$ we use $G - S$ to
denote the graph obtained by deleting $S$ from $G$, and $G[S]$ to denote the
subgraph of $G$ \emph{induced} by $S$. An \emph{isolated} vertex is not adjacent
to any other vertex. A \emph{pendant} is a vertex of degree one.

The \emph{average degree} of graph $G$, denoted by $avg\_deg(G)$, is defined as
$\frac{1}{|V(G)|} \cdot \sum_{v \in V(G)} deg(v)$. It is easy to see that
$avg\_deg(G) = 2 |E(G)|/V(G)$. A graph $G$ is called \emph{$d$-degenerate} if
every subgraph of $G$ contains a vertex of degree at most $d$. If $G$ is a
$d$-degenerate graph then $|E(G)| \le d |V(G)|$.

We say that a vertex set $S$ \emph{dominates} another vertex set $W$ if $W$ is
contained in $N[S]$. A \emph{dominating set} for a graph is a set of vertices
that dominates the entire vertex set of the graph. More formally, a vertex
subset $S$ is a \emph{dominating set} of graph $G$ if $V(G) = N[S]$.
Equivalently, a dominating set of $G$ is a set of vertices $S$ such that every
vertex which is not in $S$ is adjacent with some vertex in $S$. In the decision
version of the \DomSet problem, the input consists of a graph \(G\) and an
integer \(k\), and the task is to determine whether \(G\) has a dominating set
of size at most \(k\). We work with a generalized version of this problem. In
this version, we are given a graph $G$, an integer $k$ and a bi-partition $(V_b,
V_r)$ of vertex set $V(G)$. The objective is to find, if it exists, a vertex
subset of size at most $k$ which dominates all the vertices in $V_b$. Formally,
the task is to determine whether there exists a vertex set $S \subseteq V_b \cup
V_r$ such that $|S| \le k$ and $V_b \subseteq N[S]$. The sets $V_b, V_r$ can be
thought of as vertices with colours \emph{blue} and \emph{red}, respectively.
Vertices coloured blue are yet to be dominated while red vertices have already
been dominated. Note that the classical \DomSet problem reduces to this version
by setting $V_b = V(G)$ and $V_r = \emptyset$.

\subsection{Some Concepts from the Lossy Kernelization Framework}
\label{subsec:lossy}

This section contains a brief overview of those concepts from the lossy
kernelization framework that we need to describe our work. We encourage the
reader to see~\cite{LokshtanovPRS17-lossy} for a more comprehensive discussion
of this framework. We start with the definition of a {\em parameterized
  optimization (maximization/minimization) problem}, which is the parameterized
analog of an optimization problem in the theory of approximation algorithms.

\begin{definition} A parameterized optimization (maximization / minimization) problem is a computable function $\Pi: \Sigma^* \times \mathbb{N} \times \Sigma^* \mapsto \mathbb{R} \cup \{\pm \infty\}$.
\end{definition}

The instances of $\Pi$ are pairs $(I,k) \in \Sigma^* \times \mathbb{N}$ and a
solution to $(I,k)$ is simply a string $S \in \Sigma^*$ such that
$|S| \leq |I|+k$. The {\em value} of a solution $S$ is $\Pi(I,k,S)$. In this
article, we work with a \emph{parameterized minimization problem};
\emph{parameterized maximization problems} can be dealt with in a similar
way. 
The {\em optimum value} of $(I,k)$ is defined as:
$\textsc{OPT}_{\Pi}(I, k)= \min_{S \in \Sigma^*,\, |S| \leq |I|+k} \Pi(I,k,S)$,
and an {\em optimum solution} for $(I,k)$ is a solution $S$ such that
$\Pi(I,k,S)=\textsc{OPT}_{\Pi}(I, k)$. For a constant $c \ge 1$, $S$ is a
\emph{$c$-factor approximate} solution for $(I,k)$ if
$\frac{\Pi(I, k, S)}{OPT_{\Pi}(I, k)} \le c$. We omit the subscript $\Pi$ if the
problem being discussed is clear from the context. We now formally define an
$\alpha$-approximate reduction rule.

\begin{definition} Let $\alpha \ge 1$ be a real number and $\Pi$ be a
  parameterized minimization problem. An {\em $\alpha$-approximate reduction
    rule} is defined as a pair of polynomial-time algorithms, called the {\em
    reduction algorithm} and the {\em solution lifting algorithm}, that satisfy
  the following properties.
\begin{itemize}
\item Given an instance $(I,k)$ of $\Pi$, the reduction algorithm computes an instance $(I',k')$ of $\Pi$. 
\item Given the instances $(I,k)$ and $(I',k')$ of $\Pi$, and a solution $S'$ to $(I',k')$, the solution lifting algorithm computes a solution $S$ to $(I,k)$ such that $\frac{\Pi(I,k,S)}{\textsc{OPT}(I,k)} \leq \alpha \cdot \frac{\Pi(I',k',S')}{\textsc{OPT}(I',k')}$.
\end{itemize}
\end{definition}


\begin{definition}
  An $\alpha$-approximate reduction rule is said to be {\em strict} if for every
  instance $(I, k)$, reduced instance $(I', k')$ (which is obtained by applying
  reduction algorithm on $(I, k)$) and solution $S'$ to $(I, k)$, the solution
  lifting algorithm produces a solution $S$ to $(I, k)$ which satisfies
  $\frac{\Pi(I,k,s)}{\textsc{OPT}(I,k)} \leq \max \{
  \frac{\Pi(I',k',s')}{\textsc{OPT}(I',k')} , \alpha \}$.
\end{definition}



\begin{definition}
  An {\em $\alpha$-approximate kernel 
  } for $\Pi$ is an $\alpha$-approximate reduction rule such that the size
  \((|I'| + k')\) of the output instance of the reduction algorithm is upper
  bounded by a computable function of the original parameter $k$.
\end{definition}

For the sake of completeness, we present a definition of (classical)
kernelization. 

\begin{definition} Given an instance $(I, k)$ of some problem $\Pi$ as input, a
  kernelization algorithm returns an instance $(I', k')$ of $\Pi$ which
  satisfies following properties:
\begin{itemize}
\item $|I'| + k'$ is upper bounded by some computable function $f(k)$ which depends only on $k$.
\item $(I, k)$ is a \yes\ instance of $\Pi$ if and only if $(I', k')$ is a \yes\ instance of $\Pi$.
\end{itemize}
\end{definition}
For more details on parameterized complexity and kernelization we refer the
reader to the books of Downey and Fellows~\cite{DF-new}, Flum and
Grohe~\cite{flumgrohe}, Niedermeier~\cite{niedermeier2006}, and the more recent
books by Cygan et al.~\cite{saurabh-book} and Fomin et
al.~\cite{fomin2019kernelization}.

\section{$\alpha$-approximate Reduction Rules for Dominating Set}
\label{sec:lossy}

In the previous section we defined a generalized version of the \DomSet problem
in which the input is of the form $(G (V_b, V_r), k)$ where $G$ is a graph, $k$
is a positive integer and $(V_b, V_r)$ is a partition of $V(G)$. The objective
is to determine whether there exists a vertex set $S \subseteq V_b \cup V_r$ of
size at most $k$ such that $V_b \subseteq N[S]$ holds. We define an optimization
version of this generalized \DomSet problem in the following way.

$$\DS(G(V_b, V_r), k, S) = \left\{ \begin{array}{rl} \infty & \text{ if }V_b \not\subseteq N[S] \\ \min\{|S|, k + 1\} &\mbox{ otherwise} \end{array} \right.$$ 

Note that with this definition, one can treat all dominating sets of size
greater or equal to $k + 1$ as \emph{equally bad}. From the theoretical point of
view it is crucial to define the problem in this manner, as explained in the
part titled \textbf{Capping the objective function at $k + 1$} on Page~$15$ of
the arXiv version~\cite{lokshtanov2016lossy-arxiv} of the pioneering paper of
Lokshtanov et al.~\cite{LokshtanovPRS17-lossy}. However, from a practical point
of view, it is not necessary to distinguish solutions whose size lie above or
below a specific value. In fact, in most of the practical cases, the parameter
$k$ is not known in advance. For the sake of clarity, we replace $\DS(G(V_b,
V_r), k, S)$ by $\DS(G(V_b, V_r), S)$ and $\min\{|S|, k + 1\}$ by $|S|$ in the
above definition. Hence the input of the optimization version is simply $G(V_b,
V_r)$. Note that no reduction rule mentioned below changes the value of $k$. We
remark that in classical kernelization, similar reduction rules decrease the
value of $k$ but in the optimization version, we can prove the safeness without
decreasing the value of $k$. Hence, we can justify the above simplification in
the definition by assuming that the given value of $k$ is equal to $|V(G)|$
which does not change.

The minimum value taken over all valid solutions is denoted by $\OPT(G(V_b,
V_r))$. We start with the following simple reduction rules and their
corresponding solution lifting algorithms. In classical kernelization, these two
rules correspond to taking isolated vertices and the neighbours of pendant
vertices into a solution.

\begin{reduction rule}\label{rr:isolated} Let $X$ be the set of vertices in
  $V_b$ such that every vertex in $X$ is an isolated vertex in $G$. Then, move
  all vertices in $X$ from $V_b$ to $V_r$. Formally, let $V_b' = V_b \setminus
  X$ and $V_r' = V_r \cup X$. Return $G (V_b', V_r')$.
\end{reduction rule}

\begin{solution lifting}\label{sl:isolated} Let $S'$ be a solution for $G(V'_b, V'_r)$ and $X$ be the set of vertices moved from $V_b$ to $V_r$.
Return $S' \cup X$ as a solution for $G(V_b, V_r)$. 
\end{solution lifting}

\begin{lemma}\label{lemma:pendant-safe} Reduction Rule~\ref{rr:isolated} 
  and Solution Lifting Algorithm~\ref{sl:isolated} together constitute a strict
  $1$-approximate reduction rule.
\end{lemma}
\begin{proof} Since $X$ is a set of vertices in $V_b$, any dominating set of $G$
  must contain some subset of vertices which dominates $X$. As $X$ is a
  collection of isolated vertices, the only way to dominate $X$ is to include
  all of \(X\) in the solution.  This implies that
  $\OPT(G(V'_b, V'_r)) \le \OPT(G(V_b, V_r)) - |X|$ holds, since $X$ is
  contained in the set $V'_r$ in the reduced instance.

  Let $S'$ be a solution for $G(V'_b, V'_r)$. Then $V'_b \subseteq N[S']$
  holds. Since $V_b = V_b' \cup X$, the set $S' \cup X$ is a solution for
  $G(V_b, V_r)$. Hence
  $\DS(G(V_b, V_r), S' \cup X) \le |S'| + |X| \le \DS(G(V'_b, V'_r)), S') + |X|$
  holds.

Combining these two inequalities, we obtain the following inequality which concludes the proof.
$$\frac{\DS(G(V_b, V_r), S' \cup X)}{\OPT(G(V_b, V_r))} \le \frac{\DS(G(V'_b, V'_r)), S') + |X|}{\OPT(G(V'_b, V'_r)) + |X|} \le \max \Big\{\frac{\DS(G(V'_b, V'_r)), S')}{\OPT(G(V'_b, V'_r))}, 1\Big\}.$$
\end{proof}

\begin{reduction rule}\label{rr:pendant} Let $v$ be a vertex in $V_b$ such that
  \(v\) is a pendant vertex in $G$, and let $u$ be the unique neighbor of $v$ in
  $G$. Move all vertices in $N[u]$ from $V_b$ to $V_r$. Formally, let
  $V_b' = V_b \setminus N[u]$ and $V_r' = V_r \cup N[u]$. Return
  $G (V_b', V_r')$.
\end{reduction rule}

\begin{solution lifting}\label{sl:pendant} Let $S'$ be a solution for
  $G(V'_b, V'_r)$ and $u$ be the vertex mentioned in Reduction
  Rule~\ref{rr:pendant}.  Return $S' \cup \{u\}$ as a solution for
  $G(V_b, V_r)$.
\end{solution lifting}

\begin{lemma}\label{lemma:pendant-safe} Reduction Rule~\ref{rr:pendant} and
  Solution Lifting Algorithm~\ref{sl:pendant} together constitute a strict
  $1$-approximate reduction rule.
\end{lemma}
\begin{proof} Since vertex $v$ is in $V_b$, any dominating set of $G$ must
  contain a subset of vertices which dominates $v$.  As $v$ is a pendant vertex,
  the only vertices which can dominate $v$ are the unique neighbor $u$ of $v$,
  or $v$ itself.  This implies that
  $\OPT(G(V'_b, V'_r)) \le \OPT(G(V_b, V_r)) - |1|$ holds, as $N[u]$ is now part
  of $V'_r$ in the reduced instance.

  Let $S'$ be a solution for $G(V'_b, V'_r)$. Then $V'_b \subseteq N[S']$
  holds. Since $V_b = V_b' \cup N[u]$, the set $S' \cup \{u\}$ is a solution for
  $G(V_b, V_r)$. Hence
  $\DS(G(V_b, V_r), S' \cup \{u\}) \le |S'| + 1 \le \DS(G(V'_b, V'_r)), S') + 1$
  holds.

  Combining these two inequalities, we obtain the following inequality which
  concludes the proof.
$$\frac{\DS(G(V_b, V_r), S' \cup \{u\})}{\OPT(G(V_b, V_r))} \le \frac{\DS(G(V'_b, V'_r)), S') + 1}{\OPT(G(V'_b, V'_r)) + 1} \le \max \Big\{\frac{\DS(G(V'_b, V'_r)), S')}{\OPT(G(V'_b, V'_r))}, 1\Big\}.$$
\end{proof}

Note that we apply Reduction Rule~\ref{rr:isolated} only once while we keep
repeatedly applying Reduction Rule~\ref{rr:pendant} as long as there is a
pendant vertex in
$V_b$. 
It is not difficult to see that the following lemma holds.

\begin{lemma} \label{lemma:rr-time} There exists an algorithm that takes a graph
  $G(V_b, V_r)$ as input, runs in $\calO(n + m)$ time, and exhaustively applies
  Reduction Rules~\ref{rr:isolated} and ~\ref{rr:pendant}.
\end{lemma}

We now present the final reduction rule. 

\begin{reduction rule}\label{rr:d-deg} Find a subset $X$ of vertices in $V_b$
  such that there is a one-to-one function
  $\psi : X \rightarrow V_b \setminus X$ where
  \begin{enumerate}
  \item \label{prop:outside-nbr} for every $x$ in $X$, $\psi(x) \not\in N[x]$, 
  \item \label{prop:no-edges} for any two vertices \(x_{1}, x_{2}\) in $X$, the
    vertex \(\psi(x_{2})\) is not in the set \(N(x_{1})\), and,
  \item \label{prop:no-intersect} for any two distinct vertices $x_1, x_2$ in $X$, the
    set $N[\psi(x_1)] \cap N[\psi(x_2)]$ is empty.
  \end{enumerate}
  Move all vertices in $N[X] \cap V_{b}$ to $V_r$. Formally, let
  $V_b' = V_b \setminus N[X]$ and $V_r' = V_r \cup N[X]$. Return
  $G (V_b', V_r')$.
\end{reduction rule}

\begin{solution lifting}\label{sl:d-deg} Let $S'$ be a solution for $G(V'_b, V'_r)$ and $X$ be the subset of $V_b$ mentioned in Reduction Rule~\ref{rr:d-deg}. 
Return $S' \cup X$ as a solution for $G(V_b, V_r)$. 
\end{solution lifting}

\begin{lemma}\label{lemma:d-deg-safe} Reduction Rule~\ref{rr:d-deg} and Solution
  Lifting Algorithm~\ref{sl:d-deg} together constitute a $2$-approximate
  reduction rule.
\end{lemma}

\begin{proof} 
  As $V'_b \subseteq V_b$ holds, any solution for $G(V_b, V_r)$ is also a
  solution for $G(V_b', V'_r)$. This implies
  $\OPT(G(V_b', V'_r)) \le \OPT(G(V_b, V_r))$.

  Let the set \(X\) be as defined in Reduction Rule~\ref{rr:d-deg}, and let $Z$
  be the range of the function $\psi$.  Formally,
  $Z := \{z \in V(G) \setminus X\mid\ \exists x \in X \;;\; \psi(x) = z\}$.  By
  Properties~\ref{prop:outside-nbr} and~\ref{prop:no-edges} of Reduction
  Rule~\ref{rr:d-deg}, $Z \cap N[X]$ is the empty set.  Since only vertices in
  $N[X]$ have been moved from $V_b$ to $V_r$, we have that $Z \subseteq V'_b$
  holds.  Hence any solution \(S'\) for $G(V'_b, V'_r)$ contains at least one
  vertex from $N[z]$ for every vertex $z$ in $Z$.  By
  Property~\ref{prop:no-intersect} of Reduction Rule~\ref{rr:d-deg}, for any two
  vertices $z_1, z_2$ in $Z$, we have that $N[z_1] \cap N[z_2] = \emptyset$
  holds.  This implies that $|Z| \le |S'|$ holds. Since $\psi$ is a one-to-one
  function we get that $|X| = |Z| \le |S'|$ holds, and this in turn implies that
  $|S' \cup X| \le 2 |S'|$ holds.

  Combining the above inequalities we obtain the following inequality which
  concludes the proof.
$$\frac{\DS(G(V_b, V_r), S' \cup X)}{\OPT(G(V_b, V_r))} \le \frac{\DS(G(V'_b, V'_r)), S') + |X|}{\OPT(G(V'_b, V'_r)) + |X|} \le 2 \cdot \frac{\DS(G(V'_b, V'_r)), S')}{\OPT(G(V'_b, V'_r))}$$
\end{proof}

We remark that we apply this $2$-approximate reduction rule only once.

\begin{lemma} \label{lemma:rr-time} There exists an algorithm that takes a graph
  $G(V_b, V_r)$ as input, runs in time $\calO(n \log (n) + m)$ time, and applies
  Reduction Rule~\ref{rr:d-deg} once.
\end{lemma}
\begin{proof} We present a simple algorithm which implicitly finds a set $X$ and
  a function $\psi$ as described in Reduction Rule~\ref{rr:d-deg}, in a greedy
  fashion. This algorithm repeatedly finds a vertex \(x\) in the current set
  $V_b$ and an image \(z = \psi(x)\) using the heuristic described below, and
  moves \(N[x] \cap V_{b}\) into the set \(V_{r}\). The algorithm stops when it
  can no longer find both a vertex \(x\) and an image \(z\) for \(x\) using the
  heuristic. 

  The algorithm picks a vertex \(x\) in $V_b$ which is adjacent to the largest
  number of vertices in $V_b$. It then selects an image \(z=\psi(x)\) using a
  heuristic intended to block the fewest number of other vertices from being an
  image of a vertex picked in a later step. Note that if $z = \psi(x)$ for some
  $x$ then no vertex which is at distance at most two from $z$ can be an image
  for any other vertex in the set \(X\) (Property~\ref{prop:no-intersect} of
  Reduction Rule~\ref{rr:d-deg}). With a slight abuse of notation, we define the
  size of the second neighborhood of vertex $v$, denoted by $scd\_nbr(v)$ as
  $\sum_{u \in N(v)} |N(u)|$. Note that $scd\_nbr(v)$ is an upper bound on the
  number of vertices that are at distance at most two from vertex $v$. The
  algorithm chooses a vertex \(z\) such that (i) \(z\) is \emph{not} a neighbour
  of vertex \(x\), and (ii) $scd\_nbr(z)$ is the minimum. It then sets
  \(z=\psi(x)\).

  The algorithm can be made to run in the stated time bound using heap data
  structures.
\end{proof}

\section{Setup for Experiments}
\label{sec:exp-setup}

\subsection{Hardware and Code}

We ran our experiments on an Intel(R) Core(TM) i$5-2500$ CPU $@ 3.30$GHz machine
running Debian GNU/Linux~$8$ without using parallel processing.  We wrote the
code to implement the approximation algorithm and the lossy reduction rule in
the \texttt{C} language. We used GCC (v$6.3.0$) to compile the code.  We used
CPLEX (v$12.8.0.0$) Concert Technology for \Cplusplus to compute nearly-optimal
dominating sets, for the sake of comparisons\footnote{See
  Section~\ref{sec:results} for the details.}.  We wrote the code to generate
random graphs of various kinds, in the Python language~(v$3.5.3$) using
networkx~(v$1.11$).  All the code~\cite{code-repo} are available online.

\subsection{Test Cases}

We tested our hypothesis on randomly generated graphs and on graphs obtained
from the SuiteSparse Matrix Collection~\cite{ssmc-paper} (See~\cite{ssmc}).  We
chose graphs with edge count at most ten times their number of vertices.  That
is, our instances were all graphs with $ave\_deg$ at most $20$.  We divided our
test cases into the five categories described below, depending on how they were
constructed or obtained.  Test cases in each of these categories were classified
as \textbf{small}, \textbf{medium}, or \textbf{large}, depending on the number
of vertices.  For small, medium, and large graphs the number of vertices were in
the range $[1000 - 10,000]$, $[10,000-30,000]$, and $[30,000-50,000]$,
respectively.  Table~\ref{table:instances} shows the number of instances in each
category.

\vspace{0.3cm}
\noindent\textbf{Random Binomial Graphs or Erd$\H{o}$s-R$\'{e}$nyi Graphs $G(n, p)$:}
For a fixed positive integer $n$ and a positive constant $p \leq 1$, this model
constructs a graph on $n$ nodes by adding an edge between each unordered pair of
vertices with probability exactly $p$.  
Note that picking an edge is a Bernoulli's trial: an edge is picked with
probability $p$ independently of the existence of other edges.  Hence the
expected number of edges is $p \binom{n}{2}$.  We chose the value of $n$
uniformly at random from the specified range.  Instead of fixing the probability
value \(p\), we fixed the expected average degree ($avg\_deg$) of a target
graph.  We chose the value of $avg\_deg$ uniformly at random from the range
$[6-20]$.  We assigned $p = avg\_deg/n$ and generated the Erd$\H{o}$s-R$\'{e}$nyi
graph.

\vspace{0.3cm}
\noindent\textbf{Uniform Random Graphs $G(n, m)$:}
For two positive integers $n, m$, this model starts with an empty graph on $n$
vertices, and inserts $m$ edges in such a way that all possible
$\binom{\binom{n}{2}}{m}$ choices are equally likely.  Equivalently, it selects
two vertices uniformly at random and adds an edge between them if they are
different and no edge is already present.  It repeats this process until $m$
edges are added to the graph.  As in the previous case, instead of fixing $m$,
we fixed the expected average degree $(avg\_deg)$ of the graph.  Our program
selects $n$ uniformly at random from the specified range and $avg\_deg$ from
$[6 - 20]$.  The program then generates the graph from the model using
$(n, n \cdot avg\_deg/2)$ as the parameters.

\vspace{0.3cm}
\noindent\textbf{Watts-Strogatz Graphs $G(n, d, p)$:}
This model was proposed by Duncan J. Watts and Steven Strogatz
\cite{watts1998collective}.  It is used to produce graphs with small-world
properties i.e. short average path lengths and high clustering.  The model takes
three positive constants $n, d, p$ as parameters where $n, p$ are integers and
$p \le 1$.  It first creates a ring over the set $\{1, \dots, n\}$ and
constructs a vertex for each element in the ring.  Each vertex in the ring is
connected with its $d$ nearest neighbors ($d-1$ neighbors if $d$ is odd)
i.e. $d/2$ on each side.  After this, the model creates some shortcuts by
rewiring edges in the following way: for each edge $uv$ in the original graph
replace it by a new edge $uw$ with probability $p$ where $w$ is chosen uniformly
at random from existing vertices.  Our program selects $n$ uniformly at random
from the specified range.  As the number of expected edges is $ndp/2$, the
program selects $d$ from the range $[3-20]$ and $p$ from
$\{0.1, 0.2, \dots, 0.9\}$ uniformly at random.

\vspace{0.3cm}
\noindent\textbf{Random $d$-Regular Graphs $G(n , d)$:}
This model takes two positive integers $n, d$ and returns a $d$-regular graph on
$n$ vertices which does not contain a self-loop or parallel edges.  Networkx
contains an implementation of this model based on \cite{steger_wormald_1999}.
As the total number of edges is $nd/2$, our program requires that $nd$ be even.
It selects an even value for $n$ in the specified range and $d$ from $[3-20]$
uniformly at random.

\vspace{0.3cm}
\noindent\textbf{Suite Sparse Matrix Collection:} We select matrices from Suite Sparse Matrix Collection and treat them as adjacency matrices. 
Hence, we can work with only \emph{symmetric} matrices in this data set.  The
number of rows (and of columns) corresponds to the number of vertices, and the
number of non-zero elements is equal to twice the number of edges.  As we want
the edge density of resultant graph to be at most $10$, we select matrices for
which the number of non-zero vertices is at most $20$ times of the number of
rows.

\begin{table}[t]
  \begin{center}
  \begin{tabular}{|c|c|c|c|}
    \hline
    & small & medium & large \\
    \hline
    Erd$\H{o}$s-R$\'{e}$nyi Graphs &    100 &    100 &    100 \\
    Uniform Random Graphs  &    100 &    100 &    100 \\
    Watts-Strogatz Graphs & 100 &    100 &    100 \\
    $d$-regular Graphs &    100 &    100 &    100 \\
    SparseSuits Matrix Collection &    51 &    84 &    50 \\
    \hline
  \end{tabular}
  \end{center}
  \caption{First four rows denotes the number of graphs created in each category. The fifth column denotes the number of graphs obtained from SparseSuits Matrix Collection in respective category. \label{table:instances}}
\end{table}

\vspace{0.3cm} We make a few remarks about the Barabasi-Albert model which is
well known for generating random scale-free networks using a preferential
attachment mechanism.  Although one can create sparse graphs using this model, a
dominating set of such graphs tends to be a small fraction of the total number
of vertices.  The lossy reduction rules fail to obtain any improvement in such
cases.  In our experience, the lossy reduction rule yields improvement when the
degree distribution is binomial or uniform.  In the case of the Barabasi-Albert
model, the degree distribution follows a power law.  As per our expectations,
the lossy reduction rule does not yield any improvement. We performed
experiments to confirm this conjecture, and only $4\%$ of instances created
using this model showed any improvement.


\section{Our Experiments and Results}
\label{sec:results}

We implemented the approximation algorithm of Jones et
al.~\cite{jones2017parameterized} for \DomSet on $d$-degenerate graphs, as our
``drop-in'' approximation algorithm. Jones et al. show that if an input graph is
$d$-degenerate then their algorithm outputs a dominating set whose size is at
most $\calO(d^2)$ time that of the optimum. We note that their algorithm can in
fact take a graph $G$ and a bipartition $(V_b, V_r)$ of \(V(G)\) as input and
find a set $S$ which dominates $V_b$, such that the set returned by the
algorithm is $\calO(d^2)$ times the size of the smallest subset of $G$ which
dominates $V_b$. For every graph $G$ in the test cases, we executed the code
which performs the \emph{two} experiments described below. To ensure that the
simple and well-known Reduction Rules~\ref{rr:isolated} and \ref{rr:pendant} do
not overshadow the improvement by our $2$-approximate lossy reduction rule, we
applied these reduction rules in both the experiments.

\vspace{0.3cm}
\noindent \textbf{Exp-\texttt{AA}} (``Approximation Algorithm''): For an input
graph $G$, the code applies Reduction Rule~\ref{rr:isolated} once and Reduction
Rule~\ref{rr:pendant} exhaustively assuming $V_b = V(G)$ and $V_r = \emptyset$
as the initial partition of $V(G)$. It then finds an approximate solution for
the resulting instance using the approximation algorithm of Jones et al. The
code then constructs a solution for the original instance using the Solution
Lifting Algorithms~\ref{sl:isolated} and \ref{sl:pendant}, and returns this
solution.

\vspace{0.3cm}
\noindent \textbf{Exp-\texttt{LA}} (``Lossy Approximation''): For an input graph
$G$, the code applies Reduction Rule~\ref{rr:isolated} once and Reduction
Rule~\ref{rr:pendant} exhaustively assuming $V_b = V(G)$ and $V_r = \emptyset$
as the initial partition of $V(G)$. It then applies the \(2\)-approximate
Reduction Rule~\ref{rr:d-deg} \emph{once} to obtain another instance, say
$G(V_b', V_r')$. Let $S'$ be the set returned by the approximation algorithm of
Jones et al. when given $G(V_b', V_r')$ as input. Moreover, let $S$ be the
solution for $G(V(G), \emptyset)$ constructed by the Solution Lifting
Algorithms~\ref{sl:isolated}, \ref{sl:pendant}, and \ref{sl:d-deg} starting from
\(S'\). The algorithm returns $S$ as a solution for the original
instance. \vspace{0.3cm}

For a graph $G$, we use $\texttt{AA}(G)$ and $\texttt{LA}(G)$ to denote the
cardinalities of the solutions returned by \textbf{Exp-\texttt{AA}} and
\textbf{Exp-\texttt{LA}}, respectively. We use $\texttt{EX}(G)$ to denote the
size of the dominating set obtained using CPLEX. We ran the CPLEX program for
$30$ seconds on each input. In our experience, CPLEX narrows down the gap
between upper and lower limits for the optimum solution to less than $0.5\%$ of
a lower limit within this time.  We ran our experiments on $1385$ instances.

We now briefly discuss the running time to compute $\texttt{AA}(G)$ and
$\texttt{LA}(G)$. As expected the time required to compute the latter is
(slightly) higher, but these two running times are comparable. Consider the
$100$ large Erd$\H{o}$s-R$\'{e}$nyi Graphs instances that we created. Our machine took
$15.503$s and $48.798$s to compute $\texttt{AA}(G)$ and $\texttt{LA}(G)$,
respectively, for \emph{all} the graphs. We observed differences of similar
magnitude in the other sets of instances. Hence we do not present a detailed
analysis of the time comparison.

Table~\ref{table:improvement} summarises the percentage of instances across
different types of instances in which $\texttt{LA}(G)$ is strictly smaller than
$\texttt{AA}(G)$. In other words, in these fraction of total instances, it is
better to find a smaller instance using the reduction rules mentioned in
Section~\ref{sec:lossy}, apply the drop-in approximation algorithm, and then
lift the solution to obtain a dominating set for original, than just applying
the approximation algorithm on the original graph.
For example, for the large Erd$\H{o}$s-R$\'{e}$nyi Graphs, in $57\%$ of the graphs, $\texttt{LA}(G)$ resulted in better solution than $\texttt{AA}(G)$.

\begin{table}[t]
  \begin{center}
  \begin{tabular}{|c|c|c|c|}
    \hline
    & small & medium & large \\
    \hline
    Erd$\H{o}$s-R$\'{e}$nyi Graphs &   61.00    & 77.00    & 57.00 \\
    Uniform Random Graphs  & 71.00 &    60.00     & 63.00 \\
    Watts-Strogatz Graphs & 79.00 &    74.00 &    75.00 \\
    $d$-regular Graphs &  60.00 &     68.00 &    76.00 \\
    SparseSuits Matrix Collection & 52.94 &    58.33    & 54.00 \\
    \hline
  \end{tabular}
  \end{center}
  \caption{Each entry denotes the percentage of instances for which
    $\texttt{LA}(G)$ is strictly smaller than $\texttt{AA}(G)$. That is, this
    is the percentage of instances for which the algorithm which incorporates
    the \(2\)-approximate reduction rule performs better than the algorithm
    which does not use this rule.  \label{table:improvement}}
\end{table}

It is natural to ask for some quantitative measures to see the actual improvement. For every instance in which the above process has shown an improvement, we measure the difference between $\texttt{AA}(G)$ and $\texttt{LA}(G)$ and compare it with $\texttt{EX}(G)$. Table~\ref{table:improvement-quantify} shows the average percentage improvement in each category. 
For example, for the large Erd$\H{o}$s-R$\'{e}$nyi Graphs, if $\texttt{LA}(G)$ results in better solution than $\texttt{AA}(G)$ then, on average, the improvement in solution is $37.11\%$ to that of exact solution.

\begin{table}[t]
  \begin{center}
  \begin{tabular}{|c|c|c|c|}
    \hline
    & small & medium & large \\
    \hline
    Erd$\H{o}$s-R$\'{e}$nyi Graphs & 15.98    & 30.52    & 37.11 \\
    Uniform Random Graphs  & 8.34 &    15.98     &23.49 \\
    Watts-Strogatz Graphs & 7.97    & 19.94 &    32.15 \\ 
    $d$-regular Graphs & 1.79 &     2.00 &    2.30 \\
    SparseSuits Matrix Collection & 12.45     & 12.92    & 4.03 \\         \hline
  \end{tabular}
  \end{center}
  \caption{Each entry denotes the average value $(\texttt{AA}(G) -
    \texttt{LA}(G)) / \texttt{EX}(G) \%$ for the graphs in which $\texttt{AA}(G)
    > \texttt{LA}(G)$ holds. The larger this value, the more the impact of using
    the \(2\)-approximate reduction. \label{table:improvement-quantify}}
\end{table}

Detailed results for all the instances can be found online at~\cite{results-data}. 

\section{Conclusion}
\label{sec:conclusion}
In this work we posit that the \(\alpha\)-approximate reductions proposed by
Lokshtanov et al.~\cite{LokshtanovPRS17-lossy} as part of their pioneering work
on the Lossy Kernelization framework, can in fact be used to derive good
heuristics for obtaining better approximation algorithms. To support this thesis
we derive a \(2\)-approximate reduction for the \DomSet problem, and design an
algorithm which uses this rule as a heuristic to obtain better approximation
ratios for \DomSet. We implement this algorithm and run it on a wide variety of
benchmark instances. We demonstrate that our algorithm obtains smaller
dominating sets for a significant fraction of these benchmarks, as compared to a
state-of-the-art approximation algorithm for \DomSet.

We believe that \(\alpha\)-approximate reductions hold great promise as a way of
designing good approximation heuristics for other optimization problems as well.




\bibliography{references}

\appendix

\end{document}